\documentclass{article}
\usepackage{ifthen}
\usepackage{mdwlist}
\usepackage{amsmath,amssymb,amsfonts,amsthm}
\usepackage{bm}
\usepackage{hyperref}
\usepackage{enumitem}
\usepackage{graphicx}
\usepackage{xspace}
\usepackage{verbatim}
\usepackage{algorithm}
\usepackage{algpseudocode}
\usepackage[margin=1in]{geometry}
\usepackage{color}
\usepackage{thm-restate}
\usepackage{latexsym}
\usepackage{epsfig}

\def\a{\alpha}
\def\b{\beta}

\def\d{\delta}

\renewcommand{\epsilon}{\ve}
\def\ve{\varepsilon}

\def\l{\lambda}

\newcommand{\E}{\mbox{\bf E}}
\newcommand{\Var}{\mbox{\bf Var}}

\newcommand{\pr}[2][]{\mbox{Pr}\ifthenelse{\not\equal{}{#1}}{_{#1}}{}\!\left[#2\right]}

\newcommand{\dtv}{d_{\mathrm {TV}}}

\newtheorem{theorem}{Theorem}

\newtheorem{proposition}{Proposition}

\newtheorem{claim}{Claim}
\newtheorem{corollary}{Corollary}

\newtheorem{definition}{Definition}

\newcommand{\ignore}[1]{}

\newenvironment{prevproof}[2]{\noindent {\em {Proof of {#1}~\ref{#2}:}}}{$\hfill\qed$\vskip \belowdisplayskip}

\newcommand{\bg}[1]{\medskip\noindent{\bf #1}}

\definecolor{Red}{rgb}{1,0,0}

\newcommand{\oldbound}[1]{{}}

\title{Priv'IT: {\em Priv}ate and Sample Efficient {\em I}dentity {\em T}esting}

\author {
Bryan Cai \\
EECS, MIT\\
\tt{bcai@mit.edu}
\and
Constantinos Daskalakis \\
EECS and CSAIL, MIT \\
\tt{costis@csail.mit.edu}
\and
Gautam Kamath \\
EECS and CSAIL, MIT\\
\tt{g@csail.mit.edu}
}
\begin{document}
\maketitle
\begin{abstract}
We develop differentially private hypothesis testing methods for the small sample regime. Given a sample $\cal D$ from a categorical distribution $p$ over some domain $\Sigma$, an explicitly described distribution $q$ over $\Sigma$, some privacy parameter $\epsilon$, accuracy parameter $\alpha$, and requirements $\b_{\rm I}$ and  $\b_{\rm II}$ for the type I and type II errors of our test, the goal is to distinguish between $p=q$ and $\dtv(p,q) \geq \alpha$. 

We provide theoretical bounds for the sample size $|{\cal D}|$ so that our method both satisfies $(\epsilon,0)$-differential privacy, and guarantees $\b_{\rm I}$ and  $\b_{\rm II}$ type I and type II errors. We show that differential privacy may come for free in some regimes of parameters, and we always beat the sample complexity resulting from running the $\chi^2$-test with noisy counts, or standard approaches such as repetition for endowing non-private $\chi^2$-style statistics with differential privacy guarantees. We experimentally compare  the sample complexity of our method to that of recently proposed methods for private hypothesis testing~\cite{GaboardiLRV16,KiferR17}.
\end{abstract}

\section{Introduction} \label{sec:intro}

{\em Hypothesis testing} is the age-old problem of deciding whether observations from an unknown phenomenon $p$ conform to a model $q$. Often $p$ can be viewed as a distribution over some alphabet~$\Sigma$, and the goal is to determine, using samples from $p$, whether it is equal to some model distribution $q$ or not. This type of test is the lifeblood of the scientific method and has received tremendous study in statistics since its very beginnings. Naturally, the focus has been on minimizing the number of observations from the unknown distribution $p$ that are needed to determine, with confidence, whether $p=q$ or $p \neq q$.

In several fields of research and application, however, samples may contain sensitive information about individuals; consider for example, individuals participating in some clinical study of a  disease that carries  social stigma. It may thus be crucial to guarantee that operating on the samples needed to test a statistical hypothesis protects sensitive information about the samples. This is not at odds with the goal of hypothesis testing itself, since the latter is about verifying a property of the population $p$ from which the samples are drawn, and not of the samples themselves. 

Without care, however, sensitive information about the sample might actually be divulged by statistical processing that is improperly designed. As recently exhibited, for example, it may be possible to determine whether individuals participated in a study from data that would typically be published in genome-wide association studies~\cite{HomerSRDTMPSNC08}. Motivated in part by this realization, there has been increased recent interest in developing data sharing techniques which are private~\cite{JohnsonS13,UhlerSF13,YuFSU14,SimmonsSB16}. 

Protecting privacy when computing on data has been extensively studied in several fields ranging from statistics to diverse branches of computer science including algorithms, cryptography, database theory, and machine learning; see, e.g.,~\cite{Dalenius77,AdamW89,AgrawalA01,DinurN03,Dwork08,DworkR14} and their references. A notion of privacy proposed by theoretical computer scientists which has found a lot of traction is that of {\em differential privacy}~\cite{DworkMNS06}. Roughly speaking, it requires that the output of an algorithm on two neighboring datasets $D$ and $D'$ that differ in the value of one element be statistically close. For a formal definition see Section~\ref{sec:prelim}.

Our goal in this paper is to develop tools for privately performing statistical hypothesis testing. In particular, we are interested in studying the tradeoffs between statistical accuracy, power, significance, and privacy in the sample size. To be precise, given samples from a categorical distribution $p$ over some domain $\Sigma$, an explicitly described distribution $q$ over $\Sigma$, some privacy parameter $\epsilon$, accuracy parameter $\alpha$, and requirements $\b_{\rm I}$ and  $\b_{\rm II}$ for the type I and type II errors of our test, the goal is to distinguish between $p=q$ and $\dtv(p,q)\geq\alpha$. We want that the output of our test be $(\epsilon,0)$-differentially private, and that the probability we make a type I or type II error be $\b_{\rm I}$ and  $\b_{\rm II}$ respectively. Treating these as hard constraints, we want to {\em minimize the number of samples that we draw from $p$}. 

Notice that the {\em correctness} constraint on our test pertains to whether we draw the right conclusion about how $p$ compares to $q$, while the {\em privacy} constraint pertains to whether we respect the privacy of the samples that we draw from $p$. The pertinent question is how much the privacy constraint increases the number of samples that are needed to guarantee correctness. Our main result is that privacy may come for free in certain regimes of parameters, and has a mild cost for all regimes of parameters. 

To be precise, {\em without privacy constraints}, it is well known that identity testing can be performed from~$O({\sqrt{n} \over \alpha^2}\cdot \log {1 \over \b})$ samples, where $n$ is the size of $\Sigma$ and $\b=\min\{\b_{\rm I},\b_{\rm II}\}$, and that this is tight~\cite{BatuFFKRW01,Paninski08,ValiantV14,AcharyaDK15}. Our main theoretical result is that, {\em with privacy constraints}, the number of samples that are needed is
\begin{align}
\tilde O\left(\max \left\{\frac{\sqrt{n}}{\a^2}, \frac{\sqrt{n}}{\a^{3/2}\ve}, \frac{n^{1/3}}{\a^{5/3}\ve^{2/3}}\right\} \cdot \log (1/\b) \right). \label{eq:sample complexity}
\end{align}
Our statistical test is provided in Section~\ref{sec:mainub} where the above upper bound on the number of samples that it requires is proven as Theorem~\ref{thm:mainub}. Notice that privacy comes for free when the privacy requirement $\epsilon$ is $\Omega(\sqrt{\alpha})$ -- for example when $\epsilon=10\%$ and the required statistical accuracy is $3\%$.

The precise constants sitting in the $O(\cdot)$ notation of~Eq.~\eqref{eq:sample complexity} are given in the proof of Theorem~\ref{thm:mainub}. We experimentally verify the sample efficiency of our tests by comparing them to recently proposed private statistical tests~\cite{GaboardiLRV16,KiferR17}, discussed in more detail shortly. Fixing a differential privacy and type I, type II error constraints, we compare how many samples are required by our and their methods to distinguish between hypotheses that are $\alpha=0.1$ apart in total variation distance. 
We find that different algorithms are more efficient depending on the regime and properties desired by the analyst. 
Our experiments and further discussion of the tradeoffs are presented in Section~\ref{sec:experiments}.

\paragraph{Approach.} A standard approach to turn an algorithm differentially private is to use repetition. As already mentioned above, absent differential privacy constraints, statistical tests have been provided that use an optimal~$m=O({\sqrt{n} \over \alpha^2}\cdot \log {1 \over \b})$ number of samples. A trivial way to get $(\epsilon,0)$-differential privacy using such a non-private test is to create $O(1/\epsilon)$ datasets, each comprising $m$ samples from $p$, and run the non-private test on one of these datasets, chosen randomly. It is clear that changing the value of a single element in the combined dataset may only affect the output of the test with probability at most $\epsilon$. Thus the output is $(\epsilon,0$)-differentially private; see Section~\ref{sec:trivialub} for a proof. The issue with this approach is that the total number of samples that it draws is $m/\epsilon = O({\sqrt{n} \over \epsilon \alpha^2}\cdot \log {1 \over \b})$, which is higher than our target. See Corollary~\ref{cor:trivialub}.

A different approach towards private hypothesis testing is to look deeper into the non-private tests and try to ``privatize'' them. The most sample-efficient tests are variations of the classical $\chi^2$-test. They compute the number of times, $N_i$, that element $i \in \Sigma$ appears in the sample and aggregate those counts using a statistic that equals, or is close to, the $\chi^2$-divergence between the empirical distribution defined by these counts and the hypothesis distribution $q$. They accept $q$ if the statistic is low and reject $q$ if it is high, using some threshold.

A reasonable approach to privatize such a test is to add noise, e.g. Laplace$(1/\epsilon)$ noise, to each count $N_i$, before running the test. 
It is well known that adding Laplace$(1/\ve)$ noise to a set of counts makes them differentially private, see Theorem~\ref{thm:laplace}.
However, it also increases the variance of the statistic. 
This has been noticed empirically in recent work of~\cite{GaboardiLRV16} for the $\chi^2$-test. We show that the variance of the optimal $\chi^2$-style test statistic significantly increases if we add Laplace noise to the counts, in Section~\ref{sec:laplacechisqvar}, thus increasing the sample complexity from $O(\sqrt{n})$ to $\Omega(n^{3/4})$. So this route, too, seems problematic.

A last approach towards designing differentially private tests is to exploit the distance beween the null and the alternative hypotheses. A correct test should accept the null with probability close to $1$, and reject an alternative that is $\alpha$-far from the null with probability close to $1$, but there are no requirements for correctness when the alternative is very close to the null. We could thus try to interpolate smoothly between datasets that we expect to see when sampling the null and datasets that we expect to see when sampling an alternative that is far from the null. Rather than outputting ``accept'' or ``reject'' by merely thresholding our statistic, we would like to tune the probability that we output ``reject'' based on the value of our statistic, and make it so that the ``reject'' probability is $\epsilon$-Lipschitz as a function of the dataset. Moreover, the probability should be close to $0$ on datasets that we expect to see under the null and close to $1$ on datasets that we expect to see under an alternative that is $\alpha$-far. As we show in Section~\ref{sec:chisqsens}, $\chi^2$-style statistics have high sensitivity, requiring $\omega(\sqrt{n})$ samples to be made appropriately Lipschitz.

While both the approach of adding noise to the counts, and that of turning the output of the test Lipschitz fail in isolation, our test actually goes through by intricately combining these two approaches. It has two steps:
\begin{enumerate}
\item A {\em filtering step,} whose goal is to ``reject'' when $p$ is blatantly far from $q$. This step is performed by comparing the counts $N_i$ with their expectations under $q$, after having added Laplace$(1/\epsilon)$ noise to these counts. If the noisy counts deviate from their expectation, taking into account the extra variance introduced by the noise, then we can safely ``reject.'' Moreover, because noise was added, this step is differentially private.

\item If the filtering step fails to reject, we perform a {\em statistical step.} This step just computes the $\chi^2$-style statistic from~\cite{AcharyaDK15}, {\em without adding noise to the counts.} The crucial observation is that if the filtering step does not reject, then the statistic is actually $\epsilon$-Lipschitz with respect to the counts, and thus the value of the statistic is still differentially private. We use the value of the statistic to determine the bias of a coin that outputs ``reject.''  
\end{enumerate}

Details of our test are given in Section~\ref{sec:mainub}.

\paragraph{Related Work.}
Identity testing is one of the most classical problems in statistics, where it is traditionally called hypothesis or goodness-of-fit testing, see \cite{Pearson00,Fisher35,RaoS81,Agresti11} for some classical and contemporary references.
In this field, the focus is often on asymptotic analysis, where the number of samples goes to infinity, and we wish to get a grasp on their asymptotic distributions and error exponents \cite{Agresti11,TanAW10}.
In the past twenty years, this problem has enjoyed significant interest in the theoretical computer science community (see, i.e., \cite{BatuFFKRW01,Paninski08,LeviRR13,ValiantV14,AcharyaDK15,CanonneDGR15,DiakonikolasK16, DaskalakisDK16}, and \cite{Canonne15} for a survey), where the focus has instead been on the finite sample regime, rather than asymptotics.
Specifically, the goal is to minimize the number of samples required, while still remaining computationally tractable.

A number of recent works \cite{WangLK15,GaboardiLRV16,KiferR17} (and a simultaneous work, focused on independence testing \cite{KakizakiSF17}) investigate differential privacy with the former set of goals.
In particular, their algorithms focus on fixing a desired significance (type I error) and privacy requirement, and study the asymptotic distribution of the test statistics.
On the other hand, we are the first work to apply differential privacy to the latter line of inquiry, where our goal is to minimize the number of samples required to ensure the desired significance, power and privacy.
As a point of comparison between these two worlds, we provide an empirical evaluation of our method versus their methods.

The problem of distribution \emph{estimation} (rather than testing) has also recently been studied under the lens of differential privacy \cite{DiakonikolasHS15}.
This is another classical statistics problem which has recently piqued the interest of the theoretical computer science community.
We note that the techniques required for this setting are quite different from ours, as we must deal with issues that arise from very sparsely sampled data.

\section{Preliminaries} \label{sec:prelim}
In this paper, we will focus on discrete probability distributions over $[n]$.
For a distribution $p$, we will use the notation $p_i$ to denote the mass $p$ places on symbol $i$.

\begin{definition}
The \emph{total variation distance} between $p$ and $q$ is defined as
$$\dtv(p,q) = \frac12\sum_{i \in [n]} \left|p_i - q_i\right|.$$
\end{definition}

\begin{definition}
A randomized algorithm $M$ with domain $\mathbb{N}^{n}$ is \emph{$(\ve, \d)$-differentially private} if for all $S \subseteq \mathrm{Range}(M)$ and for all pairs of inputs $D, D'$ such that $\|D - D'\|_1 \leq 1$:
$$\Pr\left[M(D) \in S\right] \leq e^{\ve}\Pr\left[M(D') \in S\right] + \d.$$ 
If $\d = 0$, the guarantee is called \emph{pure} differential privacy.
\end{definition}

In the context of distribution testing, the neighboring dataset definition corresponds to two datasets where one dataset is generated from the other by removing one sample.
Up to a factor of 2, this is equivalent to the alternative definition where one dataset is generated from the other by arbitrarily changing one sample.

\begin{definition}
An algorithm for the \emph{$(\a, \b_{\rm I},\b_{\rm II})$-identity testing} problem with respect to a (known) distribution $q$ takes $m$ samples from an (unknown) distribution $p$ and has the following guarantees:
\begin{itemize}
\item If $p = q$, then with probability at least $1-\b_{\rm I}$ it outputs ``$p = q$;''
\item If $\dtv(p,q) \geq \a$, then with probability at least $1-\b_{\rm II}$ it outputs ``$p \neq q$.''
\end{itemize}
In particular, $\b_{\rm I}$ and $\b_{\rm II}$ are the type I and type II errors of the test. Parameter $\a$ is the radius of distinguishing accuracy. Notice that, when $p$ satisfies neither of cases above, the algorithm's output may be arbitrary.
\end{definition}

We note that if an algorithm is to satisfy both these definitions, the latter condition (the \emph{correctness} property) need only be satisfied when $p$ falls into one of the two cases, while the former condition (the \emph{privacy} property) must be satisfied for \emph{all realizations} of the samples from $p$ (and in particular, for $p$ which do not fall into the two cases above).

We recall the classical Laplace mechanism, which states that applying independent Laplace noise to a set of counts is differentially private.
\begin{theorem}[Theorem 3.6 of \cite{DworkR14}]
\label{thm:laplace}
Given a set of counts $N_1, \dots, N_n$, the noised counts $(N_1 + Y_1, \dots, N_n + Y_n)$ are $(\ve, 0)$-differentially private when the $Y_i$'s are i.i.d. random variables drawn from $Laplace(1/\ve)$. 
\end{theorem}

Finally, we recall the definition of zero-concentrated differential privacy from \cite{BunS16} and its relationship to differential privacy.

\begin{definition}
A randomized algorithm $M$ with domain $\mathbb{N}^{n}$ is \emph{$\rho$-zero-concentrated differentially private} ($\rho$-zCDP) if for all pairs of inputs $D, D'$ such that $\|D - D'\|_1 \leq 1$ and all $\a \in (1, \infty)$:
$$\mathrm{D}_\alpha(M(D)||M(D')) \leq \rho \alpha,$$
where $\mathrm{D}_\alpha$ is the $\alpha$-R\'enyi divergence between the distribution of $M(D)$ and $M(D')$. 
\end{definition}

\begin{proposition}[Propositions 1.3 and 1.4 of \cite{BunS16}]\label{prop:zcdp}
If a mechanism $M_1$ satisfies $(\ve, 0)$-differential privacy, then $M_1$ satisfies $\frac{\ve^2}{2}$-zCDP.
If a mechanism $M_2$ satisfies $\rho$-zCDP, then $M_2$ satisfies $(\rho + 2\sqrt{\rho \log(1/\d)}, \d)$-differential privacy for any $\d > 0$.
\end{proposition}

\section{A Simple Upper Bound} \label{sec:trivialub}
In this section, we provide an $O\left(\frac{\sqrt{n}}{\a^2 \ve}\right)$ upper bound for the differentially private identity testing problem.
More generally, we show that if an algorithm requires a dataset of size $m$ for a decision problem, then it can be made $(\ve, 0)$-differentially private at a multiplicative cost of $1/\ve$ in the sample size.
This is a folklore result, but we include and prove it here for completeness.
\begin{theorem}
Suppose there exists an algorithm for a decision problem $P$ which succeeds with probability at least $1 - \b$ and requires a dataset of size $m$.
Then there exists an $(\ve, 0)$-differentially private algorithm for $P$ which succeeds with probability at least $\frac{4}{5}(1 - \b) + 1/10$ and requires a dataset of size $O(m/\ve)$.
\end{theorem}
\begin{proof}
First, with probability $1/5$, we flip a coin and output yes or no with equal probability.
This guarantees that we have probability at least $1/10$ of either outcome, which will allow us to satisfy the multiplicative guarantee of differential privacy.

We then draw $10/\ve$ datasets of size $m$, and solve the decision problem (non-privately) for each of them.
Finally, we select a random one of these computations and output its outcome.

The correctness follows, since we randomly choose the right answer with probability $1/10$, or with probability $4/5$, we solve the problem correctly with probability $1 - \b$.
As for privacy, we note that, if we remove a single element of the dataset, we may only change the outcome of one of these computations.
Since we pick a random computation, this is selected with probability $\ve/10$, and thus the probability of any outcome is additively shifted by at most $\ve/10$. 
Since we know the minimum probability of any output is $1/10$, this gives the desired multiplicative guarantee required for $(\ve, 0)$-differential privacy.
\end{proof}

We obtain the following corollary by noting that the tester of \cite{AcharyaDK15} (among others) requires $O(\sqrt{n}/\a^2)$ samples for identity testing.
\begin{corollary} \label{cor:trivialub}
There exists an $(\ve, 0)$-differentially private testing algorithm for the $(\a, \beta_{\rm I}, \beta_{\rm II})$-identity testing problem for any distribution $q$ which requires
$$m = O\left(\frac{\sqrt{n}}{\ve\a^2} \cdot \log(1/\b)\right)$$
samples, where $\b = \min{(\b_{\rm I}, \b_{\rm II})}$.
\end{corollary}

\section{Roadblocks to Differentially Private Testing}
In this section, we describe roadblocks which prevent two natural approaches to differentially private testing from working.

In Section~\ref{sec:laplacechisqvar}, we show that if one simply adds Laplace noise to the empirical counts of a dataset (i.e., runs the Laplace mechanism of Theorem \ref{thm:laplace}) and then attempts to run an optimal identity tester, the variance of the statistic increases dramatically, and thus results in a much larger sample complexity, even for the case of uniformity testing.
The intuition behind this phenomenon is as follows.
When performing uniformity testing in the small sample regime (when the number of samples $m$ is the square root of the domain size $n$), we will see a $(1 - o(1)) n$ elements $0$ times, $O(\sqrt{n})$ elements $1$ time, and $O(1)$ elements $2$ times.
If we add $Laplace(10)$ noise to guarantee $(0.1, 0)$-differential privacy, this obliterates the signal provided by these collision statistics, and thus many more samples are required before the signal prevails.

In Section~\ref{sec:chisqsens}, we demonstrate that $\chi^2$ statistics have high sensitivity, and thus are not naturally differentially private.
In other words, if we consider a $\chi^2$ statistic $Z$ on two datasets $D$ and $D'$ which differ in one record, $|Z(D) - Z(D')|$ may be quite large.
This implies that methods such as rescaling this statistic and interpreting it as a probability, or applying noise to the statistic, will not be differentially private until we have taken a large number of samples.

\subsection{A Laplaced $\chi^2$-statistic has large variance}
\label{sec:laplacechisqvar}
\begin{proposition}
Applying the Laplace mechanism to a dataset before applying the identity tester of \cite{AcharyaDK15} results in a significant increase in the variance, even when considering the case of uniformity.
More precisely, if we consider the statistic
$$Z'(D) = \sum_{i \in [n]} \frac{(N_i + Y_i - m/n)^2 - (N_i + Y_i)}{m/n}$$
where $N_i$ is the number of occurrences of symbol $i$ in the dataset $D$ (which is of size $Poisson(m)$) and $Y_i \sim Laplace(1/\ve)$, then
\begin{itemize}
\item If $p$ is uniform, then $\E[Z'] = \frac{2n^2}{\ve^2 m}$ and $\Var[Z'] \geq \frac{20n^3}{\ve^4m^2}$.
\item If $p$ is a particular distribution which is $\a$-far in total variation distance from uniform, then $\E[Z'] = 4m \a^2 +  \frac{2n^2}{\ve^2 m}$.
\end{itemize}
The variance of the statistic can be compared to that of the unnoised statistic, which is upper bounded by $m^2 \a^4$.
We can see that the noised statistic has larger variance until $m = \Omega(n^{3/4})$.
\end{proposition}

\begin{proof}
First, we compute the mean of $Z'$.
Note that since $|D| \sim Poisson(m)$, the $N_i$'s will be independently distributed as $Poisson(mp_i)$ (see, i.e., \cite{AcharyaDK15} for additional discussion).
\begin{align*}
\E[Z'] &= \E \biggr[ \sum\limits_{i \in [n]} \frac{(N_i + Y_i - m/n)^2 - (N_i + Y_i)}{m/n} \biggr] \\
&= \E \biggr[ \sum\limits_{i \in [n]} \frac{(N_i - m/n)^2 - N_i}{m/n} \\
&+ \sum\limits_{i \in [n]} \frac{Y_i^2 + 2Y_i(N_i - m/n) - Y_i}{m/n} \biggr] \\
&= m \cdot \chi^2(p, q) + \sum\limits_{i \in [n]} \frac{\frac{2}{\epsilon^2}}{m/n} \\
&= m \cdot \chi^2(p, q) + \frac{2n^2}{\epsilon^2 m} 
\end{align*}
In other words, the mean is a rescaling of the $\chi^2$ distance between $p$ and $q$, shifted by some constant amount.
When $p = q$, the $\chi^2$-distance between $p$ and $q$ is $0$, and the expectation is just the second term.
Focus on the case where $n$ is even, and consider $p$ such that $p_i = (1+2\a)/n$ if $i$ is even, and $(1 - 2\a)/n$ otherwise.
This is $\a$-far from uniform in total variation distance.
Furthermore, by direct calculation, $\chi^2(p,q) = 4\a^2$, and thus the expectation of $Z'$ in this case is $4m \a^2 + \frac{2n^2}{\ve^2 m}$.

Next, we examine the variance of $Z'$.
Let $\l_i = mp_i$ and $\l_i' = mq_i = m/n$. 
By a similar computation as before, we have that 
\begin{align*}
\Var[Z'] &= \sum_{i \in [n]} \frac{1}{\lambda_i'^2} \biggr[ 2\lambda_i^2 + 4\lambda_i(\lambda_i - \lambda_i')^2 \\
&+ \frac{1}{\epsilon^2}(8\lambda_i + 2(2\lambda_i - 2\lambda_i' - 1)^2) + \frac{20}{\epsilon^4} \biggr].
\end{align*}
Since all four summands of this expression are non-negative, we have that
$$\Var[Z'] \geq \frac{20}{\ve^4}\sum_{i \in [n]} \frac{1}{\l_i'^2} = \frac{20n^3}{\ve^4 m^2}.$$

If we wish to use Chebyshev's inequality to separate these two cases, we require that $\Var[Z']$ is at most the square of the mean separation.
In other words, we require that
 $$\frac{20n^3}{\ve^4 m^2} \leq m^2 \a^4,$$
or that
$$m = \Omega\left(\frac{n^{3/4}}{\ve \a}\right).$$
\end{proof}

\subsection{A $\chi^2$-statistic has high sensitivity}
\label{sec:chisqsens}
Consider the primary statistic which we use in Algorithm \ref{alg:mainub}:
$$Z(D) = \frac{1}{m\a^2} \sum_{i \in [n]} \frac{(N_i - m q_i)^2 - N_i}{mq_i}.$$
As shown in Section \ref{sec:mainub}, $\E[Z] = 0$ if $p = q$ and $\E[Z] \geq 1$ if $\dtv(p,q) \geq \a$, and the variance of $Z$ is such that these two cases can be separated with constant probability.
A natural approach is to truncate this statistic to the range $[0,1]$, interpret it as a probability and output the result of $Bernoulli(Z)$ -- if $p = q$, the result is likely to be $0$, and if $\dtv(p,q) \geq \a$, the result is likely to be $1$.
One might hope that this statistic is naturally private.
More specifically, we would like that the statistic $Z$ has low sensitivity, and does not change much if we remove a single individual.
Unfortunately, this is not the case.
We consider datasets $D, D'$, where $D'$ is identical to $D$, but with one fewer occurrence of symbol $i$.
It can be shown that the difference in $Z$ is
\begin{align*}
|Z(D) - Z(D')| &= \frac{2|N_i - mq_i - 1|}{m^2 \a^2 q_i}
\end{align*}
Letting $q$ be the uniform distribution and requiring that this is at most $\ve$ (for the sake of privacy), we have a constraint which is roughly of the form
$$\frac{2N_in}{m^2 \a^2} \leq \ve,$$
or that
$$m = \Omega\left(\frac{\sqrt{N_i}\sqrt{n}}{\ve^{0.5}\a}\right).$$

In particular, if $N_i = n^c$ for any $c > 0$, this does not achieve the desired $O(\sqrt{n})$ sample complexity.
One may observe that, if $N_i$ is this large, looking at symbol $i$ alone is sufficient to conclude $p$ is not uniform, even if the count $N_i$ had Laplace noise added.
Indeed, our main algorithm of Section \ref{sec:mainub} works in part due to our formalization and quantification of this intuition.

\section{Priv'IT: A Differentially Private Identity Tester}
\label{sec:mainub}
In this section, we prove our main testing upper bound:
\begin{theorem}
\label{thm:mainub}
There exists an $(\ve,0)$-differentially private testing algorithm for the $(\a, \b_{\rm I}, \b_{\rm II})$-identity testing problem for any distribution $q$ which requires 
$$m = \tilde O\left(\max \left\{\frac{\sqrt{n}}{\a^2}, \frac{\sqrt{n}}{\a^{3/2}\ve}, \frac{n^{1/3}}{\a^{5/3}\ve^{2/3}}\right\} \cdot \log (1/\b) \right)$$
 samples, where $\b = \min {(\b_{\rm I}, \b_{\rm II})}$.
\end{theorem}

The pseudocode for this algorithm is provided in Algorithm \ref{alg:mainub}.
We fix the constants $c_1 = 1/4$ and $c_2 = 3/40$.
For a high-level overview of our algorithm's approach, we refer the reader to the Approach paragraph in Section~\ref{sec:intro}.

\begin{algorithm}[h]
\caption{Priv'IT: A differentially private identity tester}\label{alg:mainub}
\begin{algorithmic}[1]
\State \textbf{Input:} $\ve$; an explicit distribution $q$; sample access to a distribution $p$
\State Define $\mathcal{A} \leftarrow \{i:q_i \geq c_1\a/n\}$, $\mathcal{\bar A} \leftarrow [n] \setminus \mathcal{A}$
\State Sample $Y_i \sim Laplace(2/c_2\ve)$ for all $i \in \mathcal{A}$
\If {there exists $i \in \mathcal{A}$ such that $|Y_i| \geq \frac{2}{c_2\ve}\log\left(\frac{1}{1 - (1 - c_2)^{1/|\mathcal{A}|}}\right)$}
\State \Return either ``$p \neq q$'' or ``$p = q$'' with equal probability
\EndIf 
\State Draw a multiset $S$ of $Poisson(m)$ samples from $p$
\State Let $N_i$ be the number of occurrences of the $i$th domain element in $S$
\For {$i \in \mathcal{A}$}
\If {$|N_i + Y_i - mq_i| \geq \frac{2}{c_2\ve}\log\left(\frac{1}{1 - (1 - c_2)^{1/|\mathcal{A}|}}\right) + \max\left\{ 4\sqrt{mq_i \log n },  \log n\right\}$} 
\State \Return ``$p \neq q$''
\EndIf
\EndFor
\State $Z \leftarrow \frac{2}{m\a^2}\sum_{i \in \mathcal{A}} \frac{(N_i - mq_i)^2 - N_i}{m q_i}$
\State Let $T$ be the closest value to $Z$ which is contained in the interval $[0, 1]$
\State Sample $b \sim Bernoulli(T)$
\If {$b = 1$}
\State \Return ``$p \neq q$''
\Else
\State \Return ``$p = q$''
\EndIf
\end{algorithmic}
\end{algorithm}

\begin{prevproof}{Theorem}{thm:mainub}
We will prove the theorem for the case where $\b = 1/3$, the general case follows at the cost of a multiplicative $\log (1/\b)$ in the sample complexity from a standard amplification argument.
To be more precise, we can consider splitting our dataset into $O(\log (1/\b))$ sub-datasets and run the $\b = 1/3$ test on each one independently.
We return the majority result -- since each test is correct with probability $\geq 2/3$, correctness of the overall test follows by Chernoff bound.
It remains to argue privacy -- note that a neighboring dataset will only result in a single sub-dataset being changed.
Since we take the majority result, conditioning on the result of the other sub-tests, the result on this sub-dataset will either be irrelvant to or equal to the overall output.
In the former case, any test is private, and in the latter case, we know that the individual test is $\ve$-differentially private.
Overall privacy follows by applying the law of total probability.

We require the following two claims, which give bounds on the random variables $N_i$ and $Y_i$. 
Note that, due to the fact that we draw $Poisson(m)$ samples, each $N_i \sim Poisson(m p_i)$ independently.
\begin{claim}
\label{clm:laplacetail}
$|Y_i| \leq \frac{2}{c_2\ve}\log\left(\frac{1}{1 - (1 - c_2)^{1/|\mathcal{A}|}}\right)$ simultaneously for all $i \in \mathcal{A}$ with probability exactly $1 - c_2$.
\end{claim}
\begin{proof}
The survival function of the folded Laplace distribution with parameter $2/c_2\ve$ is $\exp\left(-c_2\ve x/2\right)$, and the probability that a sample from it exceeding the value 
$\frac{2}{c_2\ve}\log\left(\frac{1}{1 - (1 - c_2)^{1/|\mathcal{A}|}}\right)$
is equal to $1 - (1 - c_2)^{1/|\mathcal{A}|}$.
The probability that probability that it does not exceed this value is $(1 - c_2)^{1/|\mathcal{A}|}$, and since the $Y_i$'s are independent, the probability that none exceeds this value is $1 - c_2$, as desired.
\end{proof}

\begin{claim}
\label{clm:poissontail}
$|N_i - m p_i| \leq \max \left\{4\sqrt{m p_i \log n}, \log n\right\}$ simultaneously for all $i \in \mathcal{A}$ with probability at least $1 - \frac{2}{n^{0.84}} - \frac{1.1}{n}$.
\end{claim}

\begin{proof}
We consider this in two cases.
Let $X$ be a $Poisson(\l)$ random variable. 
First, assume that $\l \geq e^{-3} \log n$. 
By Bennett's inequality, we have the following tail bound \cite{Pollard15, Canonne17}:
$$\Pr\left[|X - \l| \geq x\right] \leq 2 \exp\left(-\frac{x^2}{2\l} \psi\left(\frac{x}{\l}\right)\right),$$
where 
$$\psi(t) = \frac{(1+t)\log(1+t) - t}{t^2/2}.$$
Consider $x = 4\sqrt{\l \log n}$.
At this point, we have
$$\psi(x/\l) = \psi(4\sqrt{\log n / \l}) \geq \psi(4e^{3/2}) \geq 0.23.$$
Thus, 
\begin{align*}
\Pr\left[|X - \l| \geq 4\sqrt{\l \log n}\right] &\leq 2 \exp\left(-0.23 \cdot 8\log n \right) \\
&\leq 2n^{-1.84}.
\end{align*}

Now, we focus on the other case, where $\l \leq e^{-3} \log n$. 
Here, we appeal to Proposition 1 of \cite{Klar00}, which implies the following via Stirling's approximation:
$$\Pr\left[|X - \l| \geq k\l \right] \leq \frac{k}{k-1} \exp(-\l + k \l - k \l \log k).$$
We set $k \l = \log n$, giving the upper bound
$$\frac{k}{k-1} n^{1 - \log k} \leq 1.1 \cdot n^{-2}.$$

We conclude by taking a union bound over $[n]$, with the argument for each $i \in [n]$ depending on whether $\lambda = m p_i$ is large or small.
\end{proof}

We proceed with proving the two desiderata of this algorithm, correctness and privacy.

\paragraph{Correctness.} 
We use the following two properties of the statistic $Z(D)$, which rely on the condition that $m = \Omega(\sqrt{n}/\a^2)$. 
The proofs of these properties are identical to the proofs of Lemma 2 and 3 in \cite{AcharyaDK15}, and are omitted. 
\begin{claim}
If $p = q$, then $\E[Z] = 0$.
If $\dtv(p,q) \geq \a$, then $\E[Z] \geq 1$.
\end{claim}

\begin{claim}
If $p = q$, then $\Var[Z] \leq 1/1000$.
If $\dtv(p,q) \geq \a$, then $\Var[Z] \leq 1/1000 \cdot \E[Z]^2$.
\end{claim}

First, we note that, by Claim \ref{clm:laplacetail}, the probability that we return in line 5 is exactly $c_2$. 
We now consider the case where $p = q$.
We note that by Claim \ref{clm:poissontail}, the probability that we output ``$p \neq q$'' in line 10 is $o(1)$, and thus negligible.
By Chebyshev's inequality, we get that $Z \leq 1/10$ with probability at least $9/10$, and we output ``$p = q$'' with probability at least $c_2/2 +  (1-c_2) \cdot (9/10 - c_2)^2 \geq 2/3$ (note that we subtract $c_2$ from $9/10$ since we are conditioning on an event with probability $1 - c_2$, and by union bound).
Similarly, when $\dtv(p,q) \geq \a$, Chebyshev's inequality gives that $Z \geq 9/10$ with probability at least $9/10$, and therefore we output ``$p \neq q$'' with probability at least $2/3$. 

\paragraph{Privacy.} 
We will prove $(0,c_2\ve/2)$-differential privacy.
By Claim~\ref{clm:laplacetail}, the probability that we return in line 5 is exactly $c_2$.
Thus the minimum probability of any output of the algorithm is at least $c_2/2$, and therefore $(0,c_2\ve/2)$-differential privacy implies $(\ve, 0)$-differential privacy.

We first consider the possibility of rejecting in line 11. 
Consider two neighboring datasets $D$ and $D'$, which differ by $1$ in the frequency of symbol $i$. 
Coupling the randomness of the $Y_j$'s on these two datasets, the only case in which the output differs is when $Y_i$ is such that the value of $|N_i + Y_i - mq_i|$ lies on opposite sides of the threshold for the two datasets.
Since $N_i$ differs by $1$ in the two datasets, and the probability mass assigned by the PDF of $Y_i$ to any interval of length $1$ is at most $c_2\ve/4$, the probability that the outputs differ is at most $c_2\ve/4$.
Therefore, this step is $(0,c_2\ve/4)$-differentially private.

We next consider the value of $Z$ for two neighboring datasets $D$ and $D'$, where $D'$ has one fewer occurrence of symbol $i$.
We only consider the case where we have not already returned in line 11, as otherwise the value of $Z$ is irrelevant for determining the output of the algorithm.
\begin{align*}
&Z(D) - Z(D') \\
&= \frac{1}{m\alpha^2} \biggr[ \frac{(N_i - mq_i)^2 - N_i}{mq_i} - \frac{(N_i - 1 - mq_i)^2 - (N_i - 1)}{mq_i} \biggr] \\
&= \frac{1}{m\alpha^2} \biggr[ \frac{(N_i - mq_i)^2 - N_i}{mq_i} - \frac{(N_i - mq_i)^2 - 2(N_i - mq_i) + 1 - N_i + 1}{mq_i} \biggr] \\
&= \frac{2(N_i - mq_i - 1)}{m^2\alpha^2q_i}.
\end{align*}

Since we did not return in line 11, 
\begin{align*}
|N_i - mq_i| &\leq \frac{4}{c_2\ve}\log\left(\frac{1}{1 - (1 - c_2)^{1/n}}\right) + \max\left\{4\sqrt{mq_i\log n}, \log n\right\} \\
&\leq \frac{4\log (n/c_2)}{c_2\ve} + \max\left\{4\sqrt{mq_i\log n}, \log n\right\} 
\end{align*}

This implies that 
\begin{align*}
|Z(D) - Z(D')| &= \frac{2|N_i - mq_i - 1|}{m^2\alpha^2q_i} \\
&\leq \frac{2}{m^2\a^2q_i}\left(\frac{6\log (n/c_2)}{c_2\ve} + 4\sqrt{mq_i\log n} \right).
\end{align*}

We will enforce that each of these terms are at most $c_2\ve/8$. 
\begin{align*}
\frac{12 \log (n/c_2)}{m^2\a^2 q_i c_2\ve} \leq \frac{c_2\ve}{8} &\Rightarrow m \geq \sqrt{\frac{96}{c_2^2c_1}}\frac{\sqrt{n \log (n/c_2)}}{\a^{1.5}\ve} \\
\frac{8\sqrt{\log n}}{m^{1.5}\a^2 \sqrt{q_i} } \leq \frac{c_2\ve}{8} &\Rightarrow m \geq \left(\frac{64}{c_2\sqrt{c_1}}\right)^{2/3}\frac{(n\log n)^{1/3}}{\a^{5/3}\ve^{2/3}}
\end{align*}

Since both terms are at most $c_2\ve/8$, this step is $(0,c_2\ve/4)$-differentially private.
Combining with the previous step gives the desired $(0,c_2\ve/2)$-differential privacy, and thus (as argued at the beginning of the privacy section of this proof) $\ve$-pure differential privacy.
\end{prevproof}

\section{Experiments}
\label{sec:experiments}
We performed an empirical evaluation of our algorithm, \texttt{Priv'IT}, on synthetic datasets.
All experiments were performed on a laptop computer with a 2.6 GHz Intel Core i7-6700HQ CPU and 8 GB of RAM.
Significant discussion is required to provide a full comparison with prior work in this area, since performance of the algorithms varies depending on the regime.

We compared our algorithm with two recent algorithms for differentially private hypothesis testing:
\begin{enumerate}
\item The Monte Carlo Goodness of fit test with Laplace noise from \cite{GaboardiLRV16}, \texttt{MCGOF};
\item The projected Goodness of Fit test from \cite{KiferR17}, \texttt{zCDP-GOF}.
\end{enumerate}

We note that we implemented a modified version of \texttt{Priv'IT}, which differs from Algorithm~\ref{alg:mainub} in lines 14 to 21.
In particular, we instead consider a statistic
$$Z = \sum_{i \in \mathcal{A}} \frac{(N_i - m q_i)^2 - N_i}{mq_i}.$$ 
We add Laplace noise to $Z$, with scale parameter $\Theta(\Delta/\ve)$, where $\Delta$ is the sensitivity of $Z$, which guarantees $(\ve/2, 0)$-differential privacy.
Then, similar to the other algorithms, we choose a threshold for this noised statistic such that we have the desired type I error.
This algorithm can be analyzed to provide identical theoretical guarantees as Algorithm~\ref{alg:mainub}, but with the practical advantage that there are fewer parameters to tune. 

To begin our experimental evaluation, we started with uniformity testing.
Our experimental setup was as follows.
The algorithms were provided $q$ as the uniform distribution over $[n]$.
The algorithms were also provided with samples from some distribution $p$.
This (unknown) $p$ was $q$ for the case $p = q$, or a distribution which we call the ``Paninski construction'' for the case $\dtv(p,q) \geq \a$. 
The Paninski construction is a distribution where half the elements of the support have mass $(1+\a)/n$ and half have mass $(1-\a)/n$. 
We use this name for the construction as \cite{Paninski08} showed that this example is one of the hardest to distinguish from uniform: one requires $\Omega(\sqrt{n}/\a^2)$ samples to (non-privately) distinguish a random permutation of this construction from the uniform distribution.
We fixed parameters $\ve = 0.1$ and $\a = 0.1$.
In addition, recall that Proposition \ref{prop:zcdp} implies that pure differential privacy (the privacy guaranteed by \texttt{Priv'IT}) is stronger than zCDP (the privacy guaranteed by \texttt{zCDP-GOF}).
In particular, our guarantee of $\ve$-pure differential privacy implies $\ve^2/2$-zCDP.
As a result, we ran \texttt{zCDP-GOF} with a privacy parameter of $0.005$-zCDP, which is equivalent to the amount of zCDP our algorithm provides.
Our experiments were conducted on a number of different support sizes $n$, ranging from $10$ to $10600$.
For each $n$, we ran the testing algorithms with increasing sample sizes $m$ in order to discover the minimum sample size when the type I and type II errors were both empirically below $1/3$.
To determine these empirical error rates, we ran all algorithms 1000 times for each $n$ and $m$, and recorded the fraction of the time each algorithm was correct.
As the other algorithms take a parameter $\b_{\rm I}$ as a target type I error, we input $1/3$ as this parameter.

The results of our first test are provided in Figure~\ref{fig:test1}.
The x-axis indicates the support size, and the y-axis indicates the minimum number of samples required.
We plot three lines, which demonstrate the empirical number of samples required to obtain $1/3$ type I and type II error for the different algorithms.
We can see that in this case, \texttt{zCDP-GOF} is the most statistically efficient, followed by \texttt{MCGOF} and \texttt{Priv'IT}. 
\begin{figure}
\centering
\includegraphics[scale=0.45]{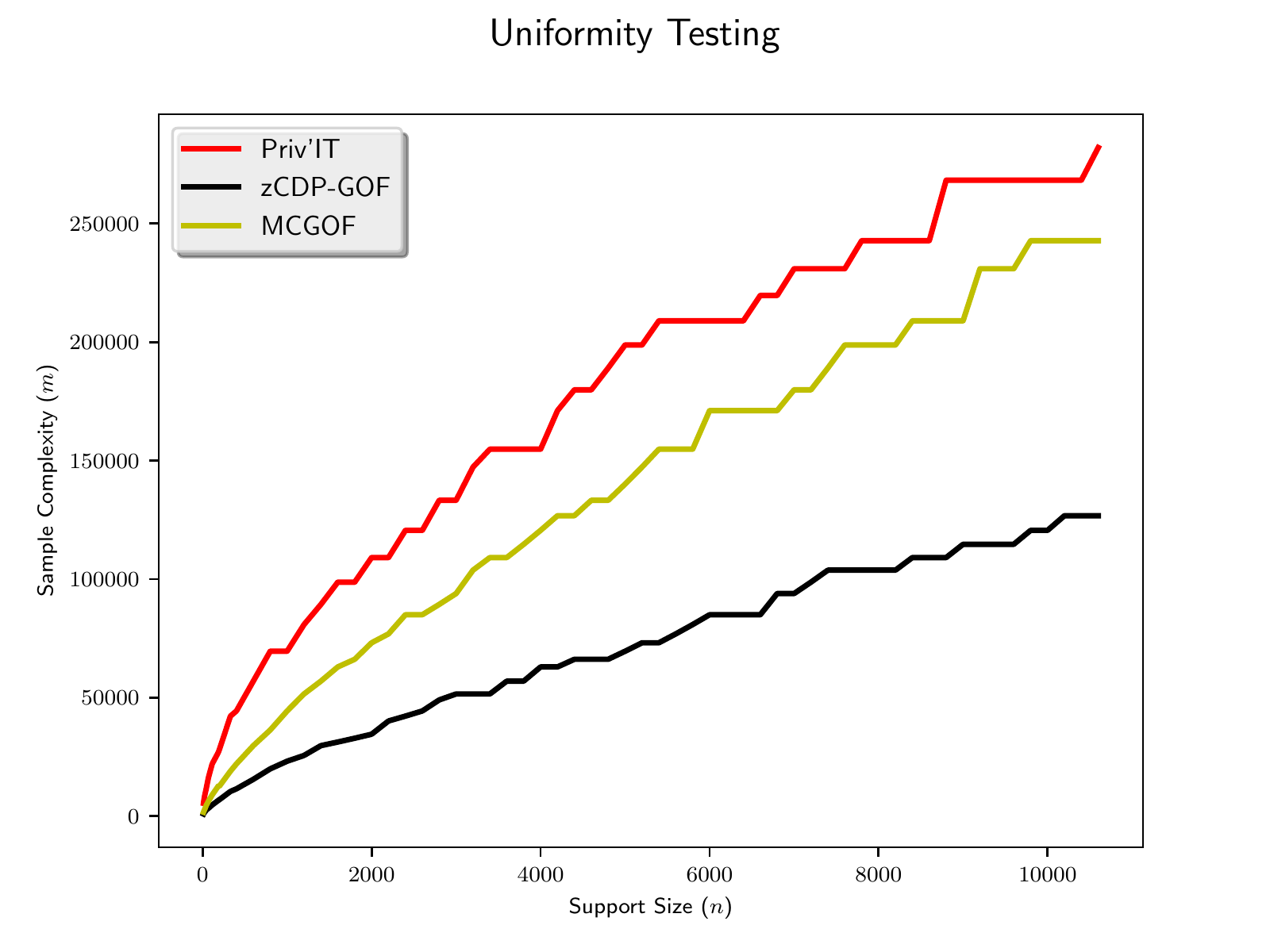}
\caption{The sample complexities of \texttt{Priv'IT}, \texttt{MCGOF}, and \texttt{zCDP-GOF} for uniformity testing}
\label{fig:test1}
\end{figure}

To explain this difference in statistical efficiency, we note that the theoretical guarantees of \texttt{Priv'IT} imply that it performs well even when data is sparsely sampled.
More precisely, one of the benefits of our tester is that it can reduce the variance induced by elements whose expected number of occurrences is less than 1.
Since none of these testers reach this regime (i.e., even \texttt{zCDP-GOF} at $n = 10000$ expects to see each element 10 times), we do not reap the benefits of \texttt{Priv'IT}.
Ideally, we would run these algorithms on the uniform distribution at sufficiently large support sizes.
However, since this is prohibitively expensive to do with thousands of repetitions (for any of these methods), we instead demonstrate the advantages of our tester on a different distribution.

Our second test is conducted with $q$ being a $2$-histogram\footnote{A $k$-histogram is a distribution where the domain can be partitioned into $k$ intervals such that the distribution is uniform over each interval.}, where all but a vanishing fraction of the probability mass is concentrated on a small, constant fraction of the support\footnote{In particular, in Figure~\ref{fig:test3}, $n/200$ support elements contained $1 - 10/n$ probability mass, but similar trends hold with modifications of these parameters.}.
This serves as our proxy for a very large support, since now we will have elements which have a sub-constant expected number of occurrences. 
The algorithms are provided with samples from a distribution $p$, which is either $q$ or a similar Paninski construction as before, where the total variation distance from $q$ is placed on the support elements containing non-negligible mass.
We ran the test on support sizes $n$ ranging from $10$ to $6800$.
All other parameters are the same as in the previous test.

The results of our second test are provided in Figure~\ref{fig:test2}.
In this case, we compare \texttt{Priv'IT} and \texttt{zCDP-GOF}, and note that our test is slightly better for all support sizes $n$, though the difference can be pronounced or diminished depending on the construction of the distribution $q$.
We found that \texttt{MCGOF} was incredibly inefficient on this construction -- even for $n = 400$ it required $130000$ samples, which is a factor of 10 worse than \texttt{zCDP-GOF} on a support of size $n = 6800$.
To explain this phenomenon, we can inspect the contribution of a single domain element $i$ to their statistic:
$$\frac{(N_i + Y_i - mq_i)^2}{mq_i}.$$
In the case where $mq_i \ll 1$ and $p = q$, this is approximately equal to $\frac{Y_i^2}{mq_i}$.
The standard deviation of this term will be of the order $\frac{1}{mq_i\ve^2}$, which can be made arbitrarily large as $mq_i \rightarrow 0$.
While \texttt{zCDP-GOF} may naively seem susceptible to this same pitfall, their projection method appears to elegantly avoid it.
\begin{figure}
\centering
\includegraphics[scale=0.45]{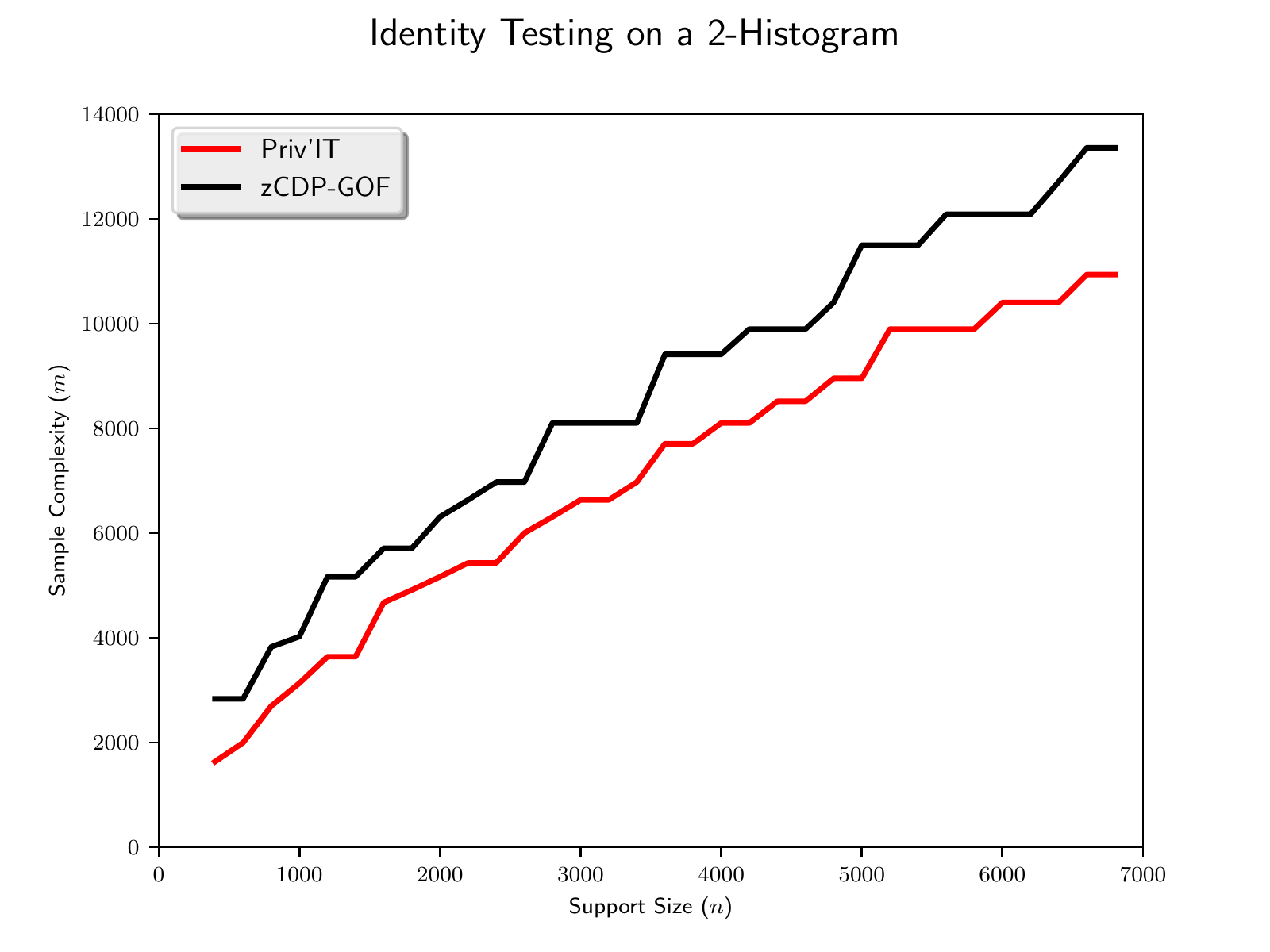}
\caption{The sample complexities of \texttt{Priv'IT} and \texttt{zCDP-GOF} for identity testing on a $2$-histogram}
\label{fig:test2}
\end{figure}

As a final test, we note that \texttt{zCDP-GOF} guarantees zCDP, while \texttt{Priv'IT} guarantees (vanilla) differential privacy.
In our previous tests, our guarantee was $\ve$-differential privacy, while theirs was $\frac{\ve^2}{2}$-zCDP: by Proposition~\ref{prop:zcdp}, our guarantees imply theirs.
In the third test, we revisit uniformity testing, but when \emph{their guarantees imply ours}.
More specifically, again with $\ve = 0.1$, we ran \texttt{zCDP-GOF} with the guarantee of $\frac{\ve^2}{2}$-zCDP and \texttt{Priv'IT} with the guarantee of $(\frac{\ve^2}{2} + \ve\sqrt{2\log(1/\d)}, \d)$ for various $\d > 0$.
We note that $\d$ is often thought in theory to be ``cryptographically small'' (such as $2^{-100}$), but we compare with a wide range of $\d$, both large and small: $\d = 1/e^t$ for $t \in \{1, 2, 4, 8, 16\}$.
This test was conducted on support sizes $n$ ranging from $10$ to $6000$.

The results of our third test are provided in Figure~\ref{fig:test3}.
We found that, for all $\d$ tested, \texttt{Priv'IT} required fewer samples than \texttt{zCDP-GOF}.
This is unsurprising for $\d$ very large and small, since the differential privacy guarantees become very easy to satisfy, but we found it to be true for even ``moderate'' values of $\d$.
This implies that if an analyst is satisfied with approximate differential privacy, she might be better off using \texttt{Priv'IT}, rather than an algorithm which guarantees zCDP. 
\begin{figure}
\centering
\includegraphics[scale=0.45]{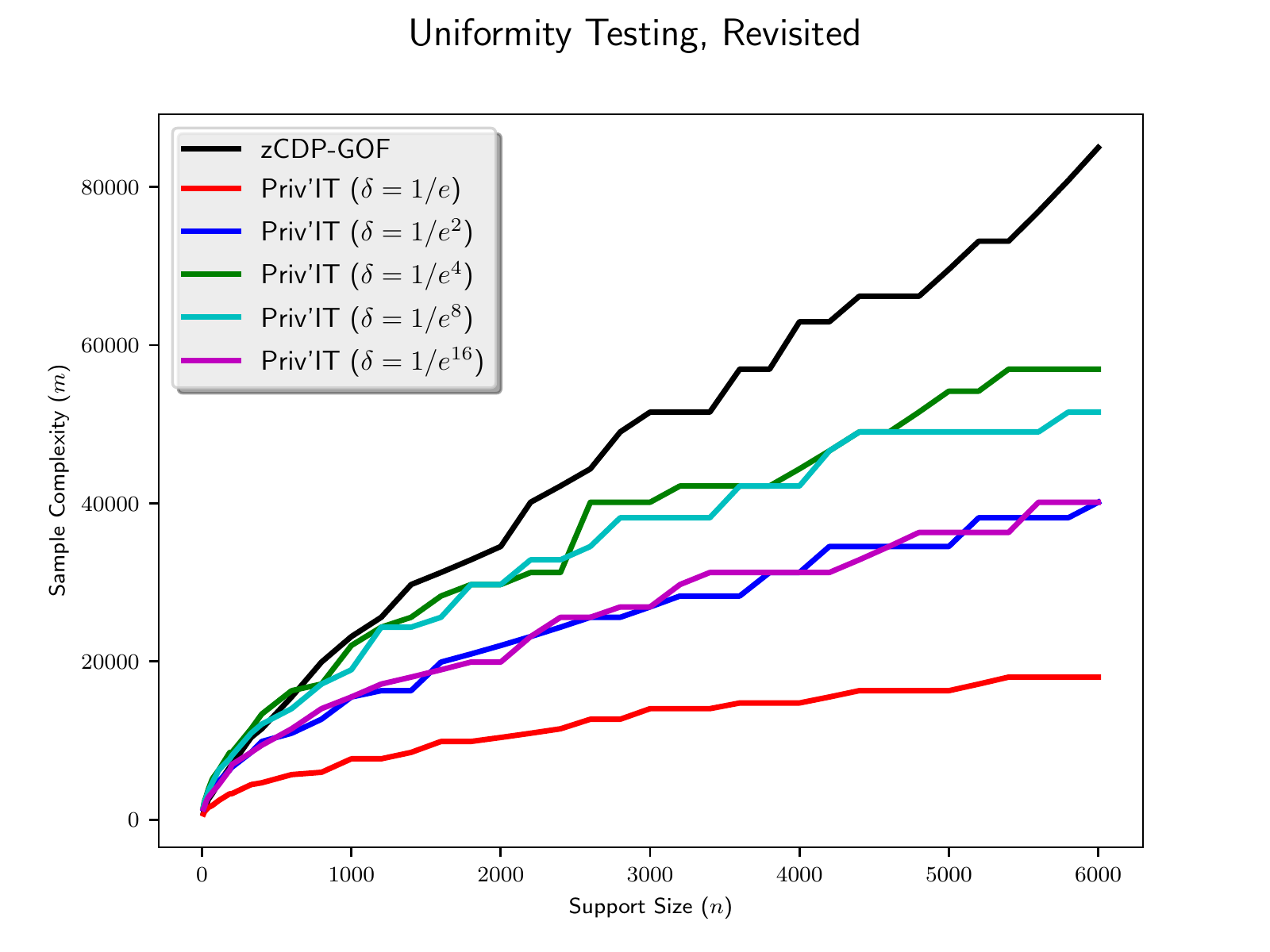}
\caption{The sample complexities of \texttt{Priv'IT} and \texttt{zCDP-GOF} for uniformity testing, with approximate differential privacy}
\label{fig:test3}
\end{figure}

While the main focus of our evaluation was statistical in nature, we will note that \texttt{Priv'IT} was more efficient in runtime than our implementation of \texttt{MCGOF}, and more efficient in memory usage than our implementation of \texttt{zCDP-GOF}.
The former point was observed by noting that, in the same amount of time, \texttt{Priv'IT} was able to reach a trial corresponding to a support size of $20000$, while \texttt{MCGOF} was only able to reach $10000$.
The latter point was observed by noting that \texttt{zCDP-GOF} ran out of memory at a support size of $11800$. 
This is likely because \texttt{zCDP-GOF} requires matrix computations on a matrix of size $O(n^2)$.  
It is plausible that all of these implementations could be made more time and memory efficient, but we found our implementations to be sufficient for the sake of our comparison.

\section*{Acknowledgments}
The authors would like to thank Jon Ullman for helpful discussions in the early stages of this work.
The authors were supported by NSF CCF-1551875, CCF-1617730, CCF-1650733, and ONR N00014-12-1-0999.
\bibliographystyle{alpha}
\bibliography{biblio}
\end{document}